\documentclass[12pt]{llncs}

\usepackage{latexsym}
\usepackage{algorithmic}
\usepackage{makeidx}
\usepackage{graphicx}

\title{Fr\'echet Distance Problems in Weighted Regions \thanks{This work was supported in part by NSF award CCF-0635013}}
\author{Yam Ki Cheung
\and Ovidiu Daescu}
\institute{Department of Computer Science, \\
University of Texas at Dallas, \\
Richardson , TX 75080, USA, \\
{\tt \{ykcheung,daescu\}@utdallas.edu}}

\begin{document}

\maketitle
\pagestyle{plain}
\begin{abstract}
We discuss two versions of the Fr\'echet distance problem in
weighted planar subdivisions.
In the first one, the distance between two points is the weighted length
of the line segment joining the points. In the second one, the distance
between two points is the length of the shortest path between the points.
In both cases, we give algorithms for finding a $(1+\epsilon)$-factor approximation of the Fr\'echet
distance between two polygonal curves.
We also consider the Fr\'echet distance between two polygonal
curves among polyhedral obstacles in $\mathcal{R}^3$
(1/ $\infty$ weighted region problem) and present a
$(1+\epsilon)$-factor
approximation algorithm.

\end{abstract}

\section{Introduction}
Measuring similarity between curves is a fundamental problem that appears in
various applications, including computer graphics and computer vision, pattern recognition,
robotics, and structural biology. One common choice for measuring the similarity between curves is the \emph{Fr\'{e}chet distance}, introduced by Fr\'echet in 1906~\cite{Fre1906}. The traditional (continuous) Fr\'echet distance $\delta_F$ for two parametric curves
$P$,$Q$: $[0,1]\rightarrow\mathcal{R}^d$ is defined as
$$\delta_F(P,Q)=\inf_{\alpha,\beta:[0,1]\rightarrow[0,1]}\sup_{r\in[0,1]} S(P(\alpha(r)),Q(\beta(r))),$$
where $\alpha$ and $\beta$ range over all continuous non-decreasing functions with $\alpha(0)=\beta(0)=0$
and $\alpha(1)=\beta(1)=1$, $S$ is a distance metric between points, and $d>0$ is the dimension of the problem.

The Fr\'echet distance is described intuitively by a man walking a dog on a leash.
The man follows a curve (path), and the dog follows the other.
Both can control their speed independently, but backtracking is not allowed.
The Fr\'echet distance between the curves is
the length of the shortest leash that is sufficient for the man and the dog to walk their paths
from start to end.

\begin{figure}[t]
    \begin{center}
    \leavevmode
    \includegraphics[height=2in]{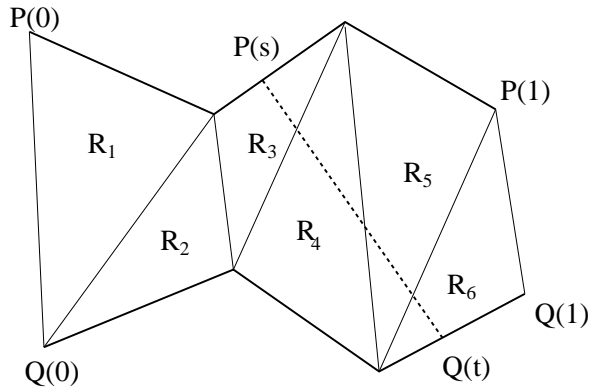}
    \caption{An example of the Fr\'echet distance problem in weighted regions.}
    \label{fig7}
    \end{center}
\end{figure}

A reparameterization is a continuous, non-decreasing surjective function.
The pair $\alpha, \beta$ of reparametrizations define how the end points of the leash, i.e.
$P(\alpha(r))$ and
$Q(\beta(r))$, sweep along their respective curves. We say that $\alpha, \beta$ is a \emph{matching} between
$P$ and $Q$ or $\alpha, \beta$ define a \emph{monotone walk} from leash $P(0)Q(0)$ to leash $P(1)Q(1)$. In this
We use $s$ and $t$ to denote $\alpha(r)$ and $\beta(r)$, respectively.
We assume $P$ and $Q$ are polygonal, and denote by $\delta(P)=(u_1,u_2,\ldots,u_p)$ and $\delta(Q)=(v_1,v_2,\ldots,v_q)$ the sequences of vertices of $P$ and $Q$, respectively.

In this paper, we study the \emph{Fr\'echet distance in weighted regions}. In this problem, a $2D$ plane is divided into a subdivision $R=\{R_1,R_2,\ldots,R_{n'} \}$ with a total of $n$ vertices, with each region
$R_i \in R$ having associated a positive integer weight $w_i$. The length of a path $\pi$ that stays within a region
$R_i$ of $R$ is defined as $w_i|\pi|$, where $|\pi|$ is the Euclidean length of $\pi$.
The length of a path in $R$ is the sum of the lengths of the subpaths within each region of $R$.

Let $P(s)Q(t)$ be the leash with endpoints $P(s)$ and $Q(t)$.
Let $R_i(P(s)Q(t))$ be the Euclidean length of
the line segment with endpoints $P(s)$ and $Q(t)$ within the region $R_i$. We define the distance between
two points $P(s)$ and
$Q(t)$ on $P$ and $Q$, respectively, as
(a) $S(P(s),Q(t))=\sum_{i=1}^{n'} w_i*R_i(P(s)Q(t))$ or
(b) $S(P(s),Q(t))$= the length of the shortest (weighted) path from $P(s)$ to $Q(t)$.

We formally define the Fr\'echet distance problem in weighted regions as the following: Given a $2D$ weighted subdivision $R$ and two
parameterized polygonal chains $P$ and $Q$ in $R$,
find the Fr\'echet distance between $P$ and $Q$, where the distance between two
points $P(s)$ and $Q(t)$ on $P$ and $Q$, respectively, is defined either as in (a) or as in (b) above.

Without loss of generality, we assume $R$ is triangulated and $P$ and $Q$ lie on boundaries of $R$ (see Fig.~\ref{fig7} as an illustration). We also discuss a special case in $\mathcal{R}^3$, where the
weights are either 1 or $\infty$ (i.e., free space and obstacles).

\subsection{Motivation} 
The motivation for studying such measurement comes from path/travel planning. Let us consider the following scenario. Suppose during a military operation, there are two teams of military units traveling on separate paths. However, before reaching their own destinations, the two teams want to maintain constant radio communications so that in case of any emergency one team can rescue the other team in time. The question is how should these two teams schedule their trips such that at any given time, the radio signal received by one team from the other is stronger than a threshold, or what is the optimal traveling schedule for each team such that minimum signal strength received is maximized.

We use the \emph{Beer-Lambert law} to model the decay (of the intensity) of an EM wave traveling through a region. Beer-Lambert law states that there is a logarithmic dependence between the \emph{transmission} $T$ of the EM wave and the product of the \emph{absorption coefficient} associated with the region $\alpha$ and the distance $l$ traveled by the EM wave in the region, i.e. $T=\frac{I}{I_o}=e^{-\alpha l}$, where $I_o$ and $I$ are the initial intensity (or power) of the EM wave and the intensity after the path, respectively \cite{Beer}. If the EM wave travels through a set of regions $R=\{R_1,R_2,\ldots,R_k\}$, the transmission $T$ can be written as $T=e^{\sum_{i=1}^{k} \alpha_i l_i }$, where $\alpha_i$ and $l_i$ are the absorption coefficient associated with $R_i$ and the distance traveled by the EM wave in $R_i$, respectively. The decay of the wave intensity can also be expressed in terms of the \emph{absorbance} $A$ which is defined as $A = -\log_{10} \frac{I}{I_0} = \sum_{i=1}^{k} \alpha_i l_i $. Notice that if we treat $R$ as a weighted subdivision, where the weight of each region $R_i$ is its absorption coefficient $\alpha_i$, $A$ is exactly the weighted length of the path traveled by the EM wave. Hence, this scheduling problem can be reduced to the Fr\'echet distance problem in weighted regions.

\subsection{Previous Work}
The Fr\'echet distance and its variants attracted considerable attention
in literature. Most previous work assumes an unweighted environment and can be divided into two categories, depending on the distance metric used. In the first category, the distance between two points is the Euclidean distance. In other words, the leash is always a line segment, as in case (a) above. Bending of the leash is not allowed. Fr\'echet distances in this category are also referred to as non-geodesic Fr\'echet distances.

Although Fr\'echet distance was introduced in 1906, it was not until 1993 that 
Alt and Godau~\cite{Alt93} gave the first polynomial-time algorithm that computes the exact non-geodesic Fr\'echet distance between two polygonal curves $P$ and $Q$. They first gave an $O(pq)$ time algorithm for solving the decision version of the problem, which is the problem to decide, for a given constant $\Delta$, whether the Fr\'echet distance between two curves is at most $\Delta$. Then, they showed how to solved optimization version of the problem, which is the problem to find the exact Fr\'echet distance, in $O(pq \log pq)$ time by applying the decision algorithm and parametric searching technique. They also presented a more practical alternative solution, which takes $O((p^2q+q^2p)\log(pq))$ time, by running binary search over a set of critical values.


Shortly after, Either and Mannila~\cite{Eiter94} introduced the concept of 
the \emph{discrete Fr\'echet distance} $\delta_{dF}$, which considers only order preserving pairing of vertices in the polygonal curves. They gave an $O(pq)$ time algorithm for computing the discrete Fr\'echet distance. They showed that the difference between the discrete Fr\'echet distance and exact Fr\'echet distance between two curves is at most the length of the longest edge of the curves. As a result, discrete Fr\'echet distance can be can be used for approximation of the exact Fr\'echet distance between two arbitrary curves. The approximation error decreases as more vertices (Steiner pionts) are introduced to the two curves.

Rote~\cite{Rote05} explored the Fr\'echet distance between more general curves. He gave an $O(pq\log pq)$ time algorithms for finding the Fr\'echet distance between piecewise smooth curves. The decision version of the problem can be answered in $O(pq)$ time. Wang et al.~\cite{Wang09} studied the problem of measuring partial similarity between curves. They presented the first exact polynomial-time algorithm to compute a partial matching between two polygonal curves, which maximizes the total length of subcurves of them that are similar to each other, where similarity is measured in Fr\'echet distance under $L_1$ or $L_\infty$ norm. Their algorithm takes $O(pq(p+q)^2\log(pq))$ time. Buchin et al.~\cite{Buc06} studied the Fr\'eechet distance between simple polygons and showed that it can be computed in polynomial time. They gave the first polynomial-time algorithm for computing
the Fr\'echet distance between simple polygons.

In the second category, the distance between two points is the geodesic distance.
Fr\'echet distances in this category are also referred to as geodesic Fr\'echet distances.
The leash is allowed to bend to achieve to the minimum length.
Maheshwari and Yi~\cite{Mah05} study the case in which the curves lie on a convex polyhedron. By constructing a visibility diagram
that encodes shortest path information for any pair of points on curves, one from each curve, they gave a $O((p^2q+pq^2)k^4\log(pqk))$ time algorithm, where $k$ is the number of faces of the polyhedron.
Cook and Wenk~\cite{cook} gave a $O(k+M^2\log(kM)\log M)$ expected time or $O(k+M^3\log(kM))$ worst-case time algorithm for computing the Fr\'echet distance between two polygonal curves when the leash is
constrained inside a polygon, where $k$ complexity of the polygon and $M$ is the larger of the complexities of the two curves. Chambers et al.~\cite{Cham08} showed how to compute the \emph{homotopic} Fr\'echet distance between two polygonal curves in the presence of obstacles in polynomial time. The homotopic Fr\'echet distance is the Fr\'echet distance with an additional continuity requirement: the leash is not allowed to switch discontinuously, e.g. jump over obstacles is forbidden unless the leash is long enough.

Weighted region shortest path problems have been investigated in computational geometry for about two decades.
Using Snell's refraction low and the continuous Dijkstra algorithm, Mitchell and Papadimitriou~\cite{Mit91}
present an $O(n^8\log\frac{nN'\rho}{\epsilon})$ time algorithm, where $N'$ is the largest integer coordinate of
vertices and $\rho$ is the ratio of the maximum weight to the minimum weight.
Aleksandrov et al.~\cite{Ale00,Ale05} provide two logarithmic discretization schemes
that place Steiner points along edges or bisectors of angles of regions in $R$,
forming a geometric progression. The placement of the Steiner points depends on an input
parameter $\epsilon >0$ and the geometry of the subdivision.
The $(1+\epsilon)$-approximation algorithms in~\cite{Ale00,Ale05} take
$O(\frac{n}{\epsilon}(\frac{1}{\sqrt{\epsilon}}+\log n)\log\frac{1}{\epsilon})$ and
$O(\frac{n}{\sqrt{\epsilon}}\log\frac{n}{\epsilon}\log\frac{1}{\epsilon})$ time, respectively.
Sun and Reif~\cite{Sun06} give an algorithm, called BUSHWHACK, which constructs a discrete graph $G$ by placing
Steiner
points along edges of the subdivision. By exploiting the geometric property of an optimal path, BUSHWHACK
computes an approximate path more efficiently as it accesses only a subgraph of $G$.
Combined with the logarithmic discretization scheme introduced in~\cite{Ale00}, BUSHWHACK takes
$O(\frac{n}{\epsilon}(\log\frac{1}{\epsilon}+\log n)\log\frac{1}{\epsilon})$ time.
Very recently, Aleksandrov et al.~\cite{Ale08} gave a query algorithm that can find an $\epsilon$-approximate
shortest path between any two points in $O(\bar{q})$ time, where $\bar{q}$ is a query time parameter. The
preprocessing time of this algorithm is
$O(\frac{(g+1)n^2}{\epsilon^{2/3}\bar{q}}\log \frac{n}{\epsilon}\log^4\frac{1}{\epsilon})$, where $g$ is the genus of the discrete graph constructed by the discretization scheme.
Cheng et al.~\cite{Cheng07} give an algorithm to approximate optimal paths in anisotropic regions,
which is a generalized case of weighted regions. Their algorithm takes $O(\frac{\rho^2\log \rho}{\epsilon^2}n^3\log(\frac{\rho n}{\rho}))$ time, where $\rho\ge 1$ is a constant such that the convex distance function of any region contains a concentric Euclidean disk with radius $1/\rho$. In weighted regions, the time complexity of the algorithm is improved to  $O(\frac{\rho\log\rho}{\epsilon}n^3\log(\frac{\rho n}{\epsilon}))$ time, where $\rho$ is the ratio of the maximum weight to the minimum weight. Very recently, Cheng et al.~\cite{Cheng07_2} also provided a query version of this algorithm that gives an approximate optimal path from a fixed source (in an anisotropic subdivision) in $O(\log \frac{\rho n}{\epsilon})$ time. The preprocessing time is $O(\frac{\rho^2n^4}{\epsilon^2}(\log \frac{\rho n}{\epsilon})^2)$.

\subsection{Definitions and Preliminaries}

We denote by $P(s)Q(t)$ the leash (line segment, also called link, or shortest path, depending of
context) from $P(s)$ to $Q(t)$, where $s,t\in [0,1]$.

Alt and Godau~\cite{Alt93} introduced the \emph{free space diagram} to solve the decision version of the
non-geodesic Fr\'echet distance problem (unweighted case): Given two polygonal curves $P$ and $Q$,
and a positive constant $\Delta$, determine if $\delta_F(P,Q)< \Delta$. A free space diagram is a
$[0,1]\times [0,1]$ parameter space such that each point $(s,t)$ in the parameter space corresponds to a leash
with end points $P(s)$ and $Q(t)$. A \emph{free space cell} $C\subseteq [0,1]^2$ is defined by two line
segments, one from each curve. $C$ corresponds to the leashes with endpoints on the two segments.
The \emph{free space} in the parameter space is defined as $\{(s,t): S(P(s),Q(t))<\Delta\}$.
That is, the free
space represents all links shorter than $\Delta$. A \emph{monotone path} is a path monotone along both
coordinate axes. There is a one-to-one correspondence between all possible
matchings of $P$ and $Q$ and all monotone paths from $(0,0)$ to $(1,1)$ in the free space diagram. Hence, $\delta_F(P,Q)< \Delta$ if and only if there exists a monotone path in the free space from the bottom left corner to the top right corner of the parameter space.

The \emph{discrete Fr\'echet distance}, introduced by Either and Mannila~\cite{Eiter94}, considers only the vertices of the (polygonal) curves. A coupling $L$ between $\delta(P)$ and $\delta(Q)$ is defined as
$$(u_{a_1},v_{b_1}),(u_{a_2},v_{b_2}),\ldots,(u_{a_m},v_{b_m}),$$
such that $a_1=1$, $b_1=1$, $a_m=p$, $b_m=q$, and for $i=1,2,\ldots,m$, we have $a_{i+1}=a_i$ or $a_{i+1}=a_i+1$, and $b_{i+1}=b_i$ or $b_{i+1}=b_i+1$. In other words, $L$ is an order preserving pairing of vertices in $P$ and $Q$. Note that each vertex can appear in $L$ more than once. The discrete Fr\'{e}chet distance is defined as $$\delta_{dF}(P,Q)=\inf_L \max_{i=1,2,\ldots,m} S(u_{a_i},v_{b_i}).$$

\subsection{Our results} Let $\epsilon >0$ be a positive constant given as part of the input. We also
assume $\epsilon < 1$ at times.
We have
the following results: (1) For weighted regions in the plane, with
$S(P(s),Q(t))=\sum_{i=1}^{n'} w_i*R_i(P(s)Q(t))$ (case (a) above), we present a $(1+\epsilon)$-approximation
algorithm that takes $O(pqN^4\log(pqN))$ time, $N=O(C(R)(\frac{n}{\epsilon}(\log\frac{1}{\epsilon}+\log n)\log\frac{1}{\epsilon})$ is the total number of Steiner points used and $C(R)$ is a constant associated with the geometry of the subdivision $R$; (2) For weighted regions in the plane, with $S(P(s),Q(t))$=the length of the shortest
(weighted) path from $P(s)$ to $Q(t)$, we present a $(1+\epsilon)$-approximation algorithm that takes
$O(C(R)^2\frac{pq}{\epsilon^2} (\log^4 \frac{1}{\epsilon})\bar{q} + \frac{(g+1)n^2}{\epsilon^{2/3}\bar{q}}\log
\frac{n}{\epsilon}\log^4\frac{1}{\epsilon})$ time, where $\bar{q}$ is a query time
parameter related to computing shortest path and $g$ is the genus of the graph constructed by the discretization scheme.
(3) We give a $(1+\epsilon)$-approximation algorithm for finding the Fr\'echet distance between $P$ and $Q$
among polyhedral obstacles in $\mathcal{R}^3$ (1 and $\infty$ weights). To the best of our knowledge this is the first
result for the $\mathcal{R}^3$ problem. The algorithm takes
$O(C(R)^2 pq (1/\epsilon^2) \log^4 (1/\epsilon)(n^2\lambda(n)\log(n/\epsilon)/{\epsilon}^4+n^2\log (n\gamma) \log(n\log\gamma)))$
time, where $\gamma$ is the ratio of the length of the longest obstacle edge to the
Euclidean distance between the two points, and $\lambda(n)$ is a very slowly-growing function related to
the inverse of the Ackermann's function.

\section{The Line Segment Leash}

In this section, we discuss the case of the line segment leash. That is, $S(P(s),Q(t))=\sum_{i=1}^{n'}w_i*R_i(P(s)Q(t))$. We first briefly discuss a pseudo polynomial exact algorithm for solving the decision version of this problem. Then, we address the optimization version of this problem directly and give a polynomial time approximation algorithm.

\subsection{An Exact Algorithm for the Decision Problem}

In this section, we give an exact algorithm extending Alt and Godau's algorithm~\cite{Alt93} to solve the decision version of the problem. Recall that the free space diagram introduced by Alt and Godau is a $[0,1]\times [0,1]$ parameter space such that each point $(s,t)$ in the parameter space corresponds to a leash with end points $P(s)$ and $Q(t)$. The free space diagram can be decomposed into $O(pq)$ free space cells such that each cell $C$ corresponds to the leashes with endpoints on two segments, one from each curve, where $p$ and $q$ are the number of vertices of $P$ and $Q$, respectively. To determine if $\delta_F(P,Q)< \Delta$ for some positive constant $\Delta$, we proceed as follows: \\

\noindent 1. Partition each cell $C$ into regions such that each region corresponds to leashes intersecting the same sequence of edges in $R$.\\
2. For each region, find the free space with respect to the positive constant $\Delta$ by solving an $O(n)$ order bivariate polynomial system with $O(n)$ equations and inequalities. \\
3. Determine if there exists a monotone path in the free space from $(0,0)$ to $(1,1)$. \\

\begin{figure}[t]
    \begin{center}
    \leavevmode
    \includegraphics[height=2in]{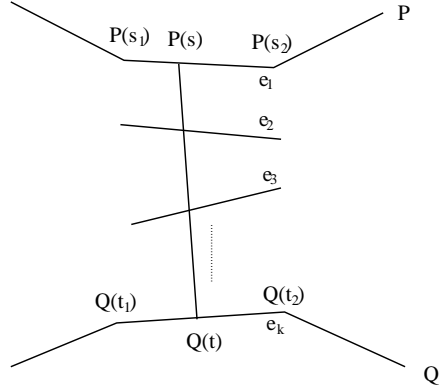}
    \caption{A leash intersects with a sequence of edges.}
    \label{seq}
    \end{center}
    \vspace*{-0.25in}
\end{figure}

The partition in step $1$ is an arrangement of $O(n)$ curves of the form $st+c_1s+c_2t+c_3=0$, where $c_1,c_2$, and $c_3$ are constants. In the worst case, the total number of regions in the parameter space is $(pqn^2)$. Each region can be defined mathematically by $O(n)$ inequalities. Once the free space in each region is computed, step $3$ can be solved by decomposing the parameter space and constructing a directed shortest path graph. We will address steps $1$ and $3$ in detail in later sections. Both steps can be solved in polynomial time. 
Next, we will show that Step $2$ can be solved, which dominates the complexity of this algorithm.

For each region, let $f(s,t)$ be the weighted length of the leash $P(s)Q(t)$ in terms of $s$ and $t$, let and $\{e_1,e_2,\ldots, e_k\}$ be the sequence of edges intersected by the leash. Note that $k=O(n)$ and $e_1$ and $e_k$ are the two segments in $P$ and $Q$, respectively.
Following~\cite{Chen01}, the weighted length of leashes within the same region have the same functional expression, i.e.
$S(P(s),Q(t))=\sqrt{1+m^2}\sum^{k-1}_{i=1}w_i|\frac{p_{i+1}-p}{m-m_{i+1}}-\frac{p_i-p}{m-m_i}| = \sqrt{1+m^2}\sum^{k}_{i=1} \frac{a_ip+b_i}{m-m_i}$, where $a_i,b_i$ are constants, $w_i$ is the weight of the
region bounded by edges $e_i$ and $e_{i+1}$, $m$, $p$ are the slope and intercept of the leash, respectively, and $m_i$, $p_i$ are the slope and intercept of $e_i\in\{e_1,e_2,\ldots, e_k\}$, respectively. Let the end points of $e_1$ be $P(s_1) = (x_1,y_1)$ and $P(s_2)=(x_2,y_2)$, respectively, and the end points of $e_k$ be $Q(t_1) = (x_3,y_3)$ and $Q(t_2) = (x_4,y_4)$, respectively. For any $s \in[s_1, s_2]$, we have\\
$$P(s) = (x_1+ (x_2-x_1)(s-s_1)/(s_2-s_1), y_1+ (y_2-y_1)(s-s_1)/(s_2-s_1)),$$
and similarly for any $t \in [t_1,t_2]$, we have\\
$$Q(t) = (x_3+ (x_4-x_3)(t-t_1)/(t_2-t_1), y_3+ (y_4-y_3)(t-t_1)/(t_2-t_1)).$$
Expressing $m$ and $p$ in terms of $s$ and $t$ , we obtain that\\
$$m =\frac{c_1s+c_2t+c_3}{d_1s+d_2t+d_3}$$ and $$p=\frac{c_1'st+c_2's+c_3't+c_4'}{d_1s+d_2t+d_3},$$\\
where $c_1,c_2,c_3,c_1',c_2',c_3',c_4',d_1,d_2$, and $d_3$ are constants. It follows that \\
$f(s,t) = \sqrt{1+(\frac{c_1s+c_2t+c_3}{d_1s+d_2t+d_3})^2}\sum^{k}_{i=1}\frac{a_i(c_1'st+c_2's+c_3't+c_4')/(d_1s+d_2t+d_3)+b_i}{(c_1s+c_2t+c_3)/(d_1s+d_2t+d_3)-m_i}.$\\


We can find the boundary of free space in this region by solving the following equation system:
$f(s,t) = \Delta$ subjects to $O(n)$ inequalities that define the region.

Since there are $O(n)$ fractional terms in $f(s,t)$, we can convert this equation system into a bivariate polynomial system of degree $O(n)$ which consists of one polynomial equation and $O(n)$ inequalities and requires pseudo polynomial time to solve~\cite{Collins75}.

We use Collins' cylindrical algebraic decomposition technique~\cite{Collins75} analyzing the polynomial system, which takes $O((\bar{m}\bar{n})^{k^r}d^k)$ time, where $\bar{m}$ is the number of equation and inequalities in the polynomial system, $\bar{n}$ is the maximum degree of polynomials, $r$ is the number of variables, $k$ is a constant and $d$ is the maximum bit length of coefficients. Hence, step $2$ takes $O(pq n^2 n^{O(1)}d^{O(1)})$ time.

This exact algorithm is not the one which ensures practical application.
Next, we are going propose an approximation algorithm for the optimization version of this problem. To approximate the Fr\'echet distance, we discretize continuous space by placing Steiner points on edges of $R$. We place the Steiner points in such a way that all links can be grouped into sets. In each set, the difference between weighted lengths of any two links is small. Instead of finding the optimal matching, we find the sequence of sets that the leash traversed in a specific matching. We then pick an arbitrary link from each set in the sequence, the maximum link length gives an approximation of the Fr\'echet distance.

\subsection{Discretization Using Steiner Points}

We first discretize the continuous space by placing Steiner points in $R$ extending the discretization scheme given in ~\cite{Ale00}.
Let $E$ be the set of all edges in $R$. Let $V$ be the set of vertices in $R$. For any point $v$
on an edge in $E$, let $E(v)$ be the set of edges incident to $v$ and let $d(v)$ be the minimum distance
between $v$ and edges in $E \setminus E(v)$.
For each edge $e\in E$, let $d(e)=\sup\{d(v)|v\in e\}$ and let $v_e$ be the point on $e$ so that $d(v_e)=d(e)$.
For each $v\in V$, the vertex radius for $v$ is defined as $r(v)=\frac{\epsilon B}{nw_{max}(v)}$,
where $\epsilon$ is a positive real number defining the quality of the approximation, $B$ is a lower
bound on $\delta_F(P,Q)$, and $w_{max}(v)$ is the maximum weight among all weighted regions incident
to $v$. The disk of radius $r(v)$ centered at $v$ defines the vertex-vicinity of $v$. $B$ can be computed
by setting weights of all regions to the minimum weight of $R$ (assume all weights are greater than zero) and applying
the continuous Fr\'echet distance algorithm
described in~\cite{Alt93}, where $p$ and $q$ are the number of vertices of the curves $P$ and $Q$, respectively.

For each edge $e=v_1v_2$ in $E$, we place Steiner points $v_{i,1}, v_{i,2},\ldots,v_{i,k_i}$ outside
the vertex-vicinity of $v_i$, for $i=1,2$, such that $|v_iv_{i,1}|=r(v_i)$, $|v_{i,j}v_{i,j+1}|=\epsilon d(v_{i,j})$,
for $j=1,2,\ldots,k_i-1$, and $v_{i,k_i}=v_e$.
It follows from~\cite{Ale00} that the number of Steiner points placed on an edge is
$O(C(e)1/\epsilon\log 1/\epsilon)$,
where $C(e)=O(\frac{|e|}{d(e)}\log\frac{|e|}{\sqrt{r(v_1)r(v_2)}})$. Let $N$ denote the total number of
Steiner points and vertices of $R$. Since we set $r(v)=\frac{\epsilon B}{nw_{max}(v)}$ for each $v\in R$, we have $N=O(C(R)(\frac{n}{\epsilon}(\log\frac{1}{\epsilon}+\log n)\log\frac{1}{\epsilon}))$, where $C(R)=\max_{e\in R}(\frac{|e|\log|e|w_{max}(e)}{d(e)B})$ is a function associated with the geometry of $R$, and $w_{max}(e)$ is the maximum weight among all weighted regions incident to the end points of $e$.

We refer to a line segment bounded by two consecutive Steiner points on an edge of $R$ as a \emph{Steiner edge}.
An \emph{hourglass} is the union of all leashes (line segments) intersecting the same sequence of Steiner edges in the same order. Let $H$ be an hourglass defined by a sequence of Steiner edges $\{e_1,e_2,\ldots,e_k\}$,
where $e_1\in P$ and $e_k\in Q$ (See Fig.~\ref{Steiner1}).

\begin{figure}[t]
    \begin{center}
    \leavevmode
    \includegraphics[height=2in]{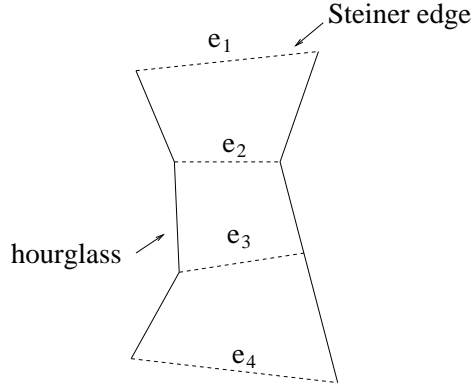}
    \caption{An hourglass is the union of leashes intersecting with the same set of Steiner edges.}
    \label{Steiner1}
    \end{center}
\end{figure}

\begin{lemma}
Let $l$ and $l'$ be two segments in $H$.  Then, $S(l)\leq (1+2\epsilon)S(l')+2\epsilon B$.
\end{lemma}
\begin{proof}
Let ${R_{i_1},R_{i_2},\ldots,R_{i_{k-1}}}$ be the weighted regions in $H$, such that $R_{i_j}$ is
between $e_j$ and $e_{j+1}$. Recall that for a region $R_k$ of $R$, $R_k(l)$ denotes the Euclidean length of $l$
within $R_k$. We have\\
$$S(l)= \sum_{j=1}^{k-1} w_{i_j} R_{i_j}(l) \leq \sum_{j=1}^{k-1} w_{i_j} (R_{i_j}(l')+{e_j}+e_{j+1} ),$$
where $w_{i_j}$ is the positive weight associated with $R_{i_j}$.

If a Steiner edge $e_j$ is outside of any vertex-vicinity, $|e_j| \leq \epsilon R_{i_{j-1}}(l')$ and
$|e_j| \leq \epsilon R_{i_j}(l')$, due to the placement of the Steiner points. If $e_j$ is incident to a vertex $v$, then $|e_j|=r(v)= ( \epsilon B) / (nw_{max}(v))$.
One segment can intersect at most $O(n)$ Steiner edges inside vertex vicinities. The result follows.
\hfill $\Box$
\end{proof}

Let $l_H$ be an arbitrary segment in $H$, with endpoints on $P$ and $Q$. Let $\alpha$ and $\beta$ be two
reparametrizations that define a
matching between $P$ and $Q$. Let $J=\{H_1,H_2,\ldots,H_k\}$ be the set of hourglasses traversed by
the leash $P(\alpha(r))Q(\beta(r))$. For an hourglass $H\in J$,
let $I_H=\{e|e=P(\alpha(r))Q(\beta(r)), r\in[0,1], e\in H\}$.
Let $H(\alpha,\beta)$ be the segment in $I_H$ with the largest weighted length.
That is, $S(H(\alpha,\beta))=\max_{e\in I_H}S(e)$.
Let $\delta(\alpha,\beta)=\sup_{r\in[0,1]} S(P(\alpha(r)),Q(\beta(r)))$.

\begin{lemma}
$|\max_{H\in J}S(l_H) - \delta(\alpha,\beta)|\leq 4\epsilon\delta(\alpha,\beta)$. \label{lemma2}
\end{lemma}

\begin{proof}
Applying $Lemma~1$, we have,
$$S(l_H)\leq (1+2\epsilon)S(H(\alpha,\beta))+2\epsilon D,$$
and
$$S(H(\alpha,\beta)) \leq (1+2\epsilon)S(l_H) + 2\epsilon D.$$
Assume $\epsilon < 1/2$, it follows,
$$\max_{H\in J}S(l_H)\geq \frac{1-2\epsilon}{1+2\epsilon}\delta(\alpha,\beta)\geq (1-4\epsilon)\delta(\alpha,\beta),$$
and
$$\max_{H\in J}S(l_H) \leq (1+2\epsilon)\max_{H\in J}S(H(\alpha,\beta))+2\epsilon D
 \leq (1+4\epsilon)\delta(\alpha,\beta).$$
\hfill $\Box$

\end{proof}

We say that $\delta(J)=\max_{H\in J}S(l_H)$ is a $4\epsilon$-approximation of $\delta(\alpha,\beta)$.
Given a sequence of hourglasses $J=\{H_1,H_2,\ldots,H_k\}$, we call $J$ a \emph{legal} sequence if there exists
two reparmetrizations $\alpha,\beta$ that define a leash traversing the same sequence of hourglasses as $J$.

\begin{lemma}
\label{lem-apr}
We can find a value $\delta'_F(P,Q)$ such that
$|\delta'_F(P,Q) - \delta_F(P,Q)| \leq 4\epsilon \delta_F(P,Q)$, that is, $\delta'_F(P,Q)$ is a
$4\epsilon$-approximation of $\delta_F(P,Q)$.
\end{lemma}

\begin{proof}
Let $\hat{\alpha}$,$\hat{\beta}$ be the optimal reparametrizations that give
the Fr\'echet distance between $P$ and $Q$. Let $\hat{J}$ be the sequence of hourglasses traversed by the leash
defined by $\alpha$ and $\beta$.
Applying $Lemma~2$, we have,
$$|\delta(\hat{J}) - \delta(\hat{\alpha},\hat{\beta})|= |\delta(\hat{J})-\delta_F(P,Q)| \leq 4\epsilon\delta(\hat{\alpha},\hat{\beta}) = 4\epsilon\delta_F(P,Q).$$
So, there is a legal sequence of hourglasses that gives a $4\epsilon$-approximation of
$\delta_F(P,Q)$.
\hfill $\Box$
\end{proof}



\subsection{Fr\'echet Distance Between Two Segments}

Next, for simplicity, we study a special case of the problem, in which each curve consists of one line segment.
See Fig.~\ref{fig5} for an illustration.
Similar to Alt and Godau's algorithm, we attack the problem in the parameter space $D=[0,1]^2$,
where a leash $P(s)Q(t)$ is associated with a point $(s,t)\in D$. We refer to $(s,t)\in D$ as the \emph{dual
point} of the leash $P(s)Q(t)$.

\begin{lemma}
All leashes through a Steiner point $v$ correspond to a curve in $D$, with
equation $C_v: st+c_1s+c_2t+c_3=0$, where $c_1$, $c_2$, and $c_3$ are constants.
\end{lemma}

\begin{proof}
Let $v=(a,b)$, $P(0)=(x_1,y_1)$, $P(1)=(x_2,y_2)$, $Q(0)=(x_3,y_3)$, $Q(1)=(x_4,y_4)$. We have
$$P(s)= (x_1+(x_2-x_1)s, y_1+ (y_2-y_1)s),$$
$$Q(t)= (x_3+(x_4-x_3)t, y_3+ (y_4-y_3)t),$$
and \\
$P(s)Q(t): y-(y_1+ (y_2-y_1)s) = \frac{y_1+ (y_2-y_1)s - (y_3+ (y_4-y_3)t)}{x_1+(x_2-x_1)s -  (x_3+(x_4-x_3)t)}(x-(x_1+(x_2-x_1)s) ).$\\
Since $P(s)Q(t)$ passes through $v$, we obtain that \\
$(b-y_1-(y_2-y_1)s)( x_1- x_3+(x_2-x_1)s - (x_4-x_3)t)) = (a-x_1-(x_2-x_1)s)(y_1-y_3 + (y_2-y_1)s -  (y_4-y_3)t)$, \\and thus we have
$$C_v: st+c_1s+c_2t+c_3=0,$$
where $c_1,c_2$, and $c_3$ are constants.
\hfill $\Box$
\end{proof}

\begin{figure}[t]
    \begin{center}
    \leavevmode
    \includegraphics[height=1.75in]{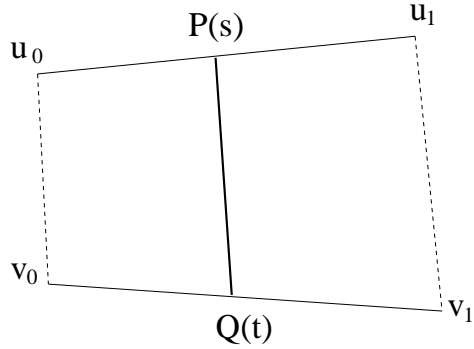}
    \caption{A special case of the Fr\'echet distance problem.}
    \label{fig5}
    \end{center}
\end{figure}

We call the curve $C_v$ the \emph{dual curve} of Steiner point $v$.
$C_v$ is continuous and monotone along both (s and t) axes.
Obviously, $v$ lies on a leash if and only if the dual point of the leash in $D$ lies on $C_v$.
Next we define the relative position of two leashes with respect to a Steiner point $v$.
Given two leashes $P(s)Q(t)$ and $P(s')Q(t')$, we say they are on the \emph{same side} of $v$ if and only if
there exists two continuous functions $\alpha', \beta': [0,1]\rightarrow[0,1]$, such that
$\alpha'(0)=s$, $\alpha'(1)=s'$, $\beta'(0)=t$, $\beta'(1)=t'$ and
$v\notin P(\alpha'(r)Q(\beta'(r)), \forall r\in[0,1]$. Note that $\alpha',\beta'$ are not necessarily
monotone. Intuitively, two leashes are on the same side of $v$, if one leash can transform to the other one by sweeping its end points on their respective curves without crossing $v$.

\begin{lemma}
$C_v$ divides $D$ into two partitions, each partition corresponding to all links on the same side of
point $v$.
\end{lemma}

\begin{proof}

Given two links $P(s)Q(t)$ and $P(s')Q(t')$, since there is a one-to-one correspondence between all possible walks from $P(s)Q(t)$ to $P(s')Q(t')$ and all paths from $(s,t)$ to $(s',t')$ in $D$, if $(s,t)$, $(s',t')$ are on different sides of $C_v$, all paths between $(s,t)$ and $(s',t')$ must pass $C_v$ at least once. That is, no walk exists between the two links without crossing $v$.
\hfill $\Box$
\end{proof}

We partition $D$ by the dual curves of the Steiner points. Using the algorithm in~\cite{AMA00}, the partition can be computed in
$O(N\log N+k)$ time and $O(N+k)$ space, where $N=O(C(R)(\frac{n}{\epsilon}(\log\frac{1}{\epsilon}+\log n)\log\frac{1}{\epsilon}))$ is the total number of Steiner points and $k$ is the number of cells,
which is $O(N^2)$ in the worst case. Let $A$ denote the partition.
\begin{lemma}
In $A$, each cell corresponds to an hourglass.
\end{lemma}
\begin{proof}
If two points $(s,t),(s',t')$ belong to the same cell, then there exists a walk between $P(s)Q(t)$ and $P(s')Q(t')$ that does not cross any Steiner point, i.e. $P(s)Q(t)$ and $P(s')Q(t')$ intersect with the same set of Steiner edges. If $(s,t), (s',t')$ belong to different cells, then any walk between $P(s)Q(t)$ and $P(s')Q(t')$ must cross at least one Steiner point. That is, $P(s)Q(t),P(s')Q(t')$ do not intersect with the same set of Steiner edges. The result follows.
\hfill $\Box$
\end{proof}

The Fr\'echet distance between $P$ and $Q$ can then be approximated as follows: \\
\noindent 1. Place Steiner points as described previously. \\
2. Partition $D$ by dual curves of all Steiner points. \\
3. For each cell $Z$, choose an arbitrary leash $l$ and assign $S(l)$ as the weight of the cell. \\
4. Find a monotone path $T$ in $D$, from its left bottom corner, $(0,0)$, to its top right corner, $(1,1)$, such
that the maximum weight of the cells traversed by $T$ is minimized. \\

Since the sequence of cells traversed by a monotone path in $D$ corresponds
to a legal sequence of hourglasses traversed by a leash following a match of $P$ and $Q$, and vice versa, by
Lemma~\ref{lem-apr} the maximum weight of the cells traversed by $T$ is a $4\epsilon$-approximation of the
Fr\'echet distance between $P$ and $Q$.

\subsection{Finding an optimal path in $D$}

We define the cost of a path between two points in $D$ as the maximum weight of the regions traversed by the
path.
We decompose $D$ by extending a horizontal as well as a vertical line from every vertex until it reaches the
boundary of $D$. See Fig.~\ref{decompose} for an illustration. Let the new subdivision be $D'$.

\begin{lemma}
Given an edge $e$ in $D'$, let $p$ be a point on $e$ and let $T_p$ be an arbitrary monotone path from
$(0,0)$, i.e. the bottom left corner of $D'$, to $p$. Then, there exists a monotone path from $(0,0)$ to any
point on $e$ with the same cost of $T_p$.
\label{monoPath}
\end{lemma}

\begin{proof}
We consider $e$ in $D$. Note that $D$ consists of dual lines of Steiner points only. In $D$, we extend a vertical line as well as a horizontal
line from the end points of $e$. Let $L_1, L_2$ be the two horizontal lines which define a
horizontal \emph{slab} and $L_3,L_4$ be the two vertical lines which define a vertical slab.
Obviously, $T_p$ has the same cost in $D$ as it does in $D'$.
Without loss of generality, assume $T_v$ enters the horizontal slab before it enters the vertical slab. Let $v$ be the entry point.
Let $S=\{s_1, s_2,\ldots,s_k\}$ be the set of curves on edges of $D$, which are bounded by the horizontal
slab and traversed by $T_p$. Since the horizontal slab does not contain any vertex of $D$ in its interior,
no two curves in $S$ intersect with each other.
Hence we can construct a monotone path from $v$ to any point on $e$ that traverses $S$, i.e.
we can contract a monotone path from $(0,0)$ to any point on $e$ such that it has the
same cost as $T_p$. See Fig.~\ref{slab} for an illustration. The same argument holds
if $T_p$ enters the vertical slab first. \hfill $\Box$

\hfill $\Box$
\end{proof}
\begin{figure}[t]
    \begin{center}
    \leavevmode
    \includegraphics[height=2in]{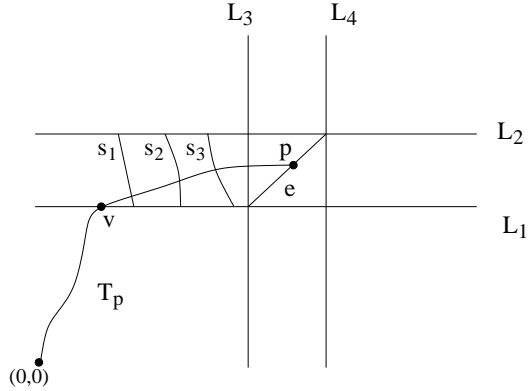}
    \caption{There exists a monotone path from $(0,0)$ to every point in $e$, which has the same cost as $T_p$.}
    \label{slab}
    \end{center}
\end{figure}

Thus, given an optimal monotone path $T_p$ from $(0,0)$ to an arbitrary point $p\in e$, we can construct a
monotone path from $(0,0)$ to any point in $e$, which has the same cost as $T_p$ and is also optimal.
To find the optimal monotone path from $(0,0)$ to $(1,1)$, we construct a directed graph $G$ having edges of $D'$ as vertices. In $G$, a directed edge is added from a node $v_1$ to a node
$v_2$ if, in $D'$, $v_1$ and $v_2$ are on the boundary of the same cell and there exists a monotone path from $v_1$ to $v_2$. The
optimal monotone path can be found by running a modified Dijkstra's algorithm on $G$. \\

\noindent {\bf Time complexity:} The complexity of $D'$ is $N^4$, where $N=O(C(R)(\frac{n}{\epsilon}(\log\frac{1}{\epsilon}+\log n)\log\frac{1}{\epsilon}))$ is
the number of Steiner points used, since each cell in $D'$ has $O(1)$ edges and vertices, and thus each
cell
contributes $O(1)$ edges to $G$. Then, $G$ has $O(N^4)$ vertices and $O(N^4)$ edges and Dijkstra's algorithm
takes $O(N^4\log N)$ time.

\subsection{Fr\'echet Distance Between Two Polygonal curves}

We extend the algorithm above to approximate the Fr\'echet distance between two polygonal chains $P$ and $Q$.
Recall
that $p$ and $q$ are the number of vertices of $P$ and $Q$, respectively.
The parameter space $D$ can be divided into $(p-1)(q-1)$ subspaces, such that each subspace corresponds to all leashes bounded by the same two segments, one
from each chain. Each subspace can be partitioned independently by introducing dual curves of Steiner points. The Fr\'echet distance $\delta_F(P,Q)$ can be approximated by finding an optimal monotone path from the bottom left
corner
to the top right corner of $D$. The complexity of $D$ is $O(pqN^2)$ and the complexity of the decomposed parameter space $D'$ is $O(p^2q^2N^4)$, where
$N=O(C(R)(\frac{n}{\epsilon}(\log\frac{1}{\epsilon}+\log n)\log\frac{1}{\epsilon}))$ is the total number of Steiner points. It takes $O(p^2q^2N^4\log(pqN))$ time to approximate $\delta_F(P,Q)$.

\begin{theorem}
An approximation of the Fr\'echet distance between two polygonal curves in a weighted subdivision can be computed in \\$O(p^2q^2N^4\log(pqN))$ time, where $N=O(C(R)(\frac{n}{\epsilon}(\log\frac{1}{\epsilon}+\log n)\log\frac{1}{\epsilon}))$ is the total number of Steiner points.
\end{theorem}

\section{Geodesic Fr\'echet Distance}

In this section we study two versions of the geodesic Fr\'echet distance problem. We do not require
the leash to be homotopic, i.e. the leash is allowed to sweep discontinuously without penalty. For example, the leash can pass through or jump over obstacles.

\subsection{Geodesic Fr\'echet Distance in Weighted Regions in $\mathcal{R}^2$}

Recall that the cost (weighted length) of a path in $R$ is defined as the weighted sum of its Euclidean
lengths within each region of $R$. A geodesic path between two points in $R$ is a path between those points that
has minimum cost. Here, the distance between two points in $R$ is the cost of the geodesic path between those
points.

We prove that the geodesic Fr\'echet distance can be approximated by the discrete geodesic Fr\'echet distance.
Given two polygonal curves $P,Q$, to approximate $\delta_F(P,Q)$, we need to add additional vertices to $P$ and
$Q$. We follow a similar approach as in~\cite{Ale00}, except that we define the vertex radius of a
vertex $v$ in $P$ or $Q$ as $r(v)=\frac{\epsilon B}{w_{max}(v)}$, where $\epsilon$ is a positive real number
defining the quality of the approximation, $B$ is a lower
bound on $\delta_F(P,Q)$, and $w_{max}(v)$ is the maximum weight of regions incident
to $v$. 
Recall
that the disk of
radius $r(v)$ centered at $v$ defines the vertex-vicinity of $v$. New vertices are placed on edges of $P$ nd $Q$
forming a geometric progression with ratios depending on $\epsilon$ and on the geometry of $R$
(see Section $3.1$ for details). The total number of additional vertices introduced on each edge is
$O(C(R)\frac{1}{\epsilon} \log^2 \frac{1}{\epsilon})$, where $C(R)$ is a constant associated with geometry of $R$. Let $P'$ and $Q'$ be the new polygonal chains and let $p'=O(C(R)\frac{p}{\epsilon} \log^2 \frac{1}{\epsilon})$ and $q'=O(C(R)\frac{q}{\epsilon} \log^2 \frac{1}{\epsilon})$ be the number of vertices of $P'$ and $Q'$,
respectively.

\begin{figure}[t]
    \begin{center}
    \leavevmode
    \includegraphics[height=2in]{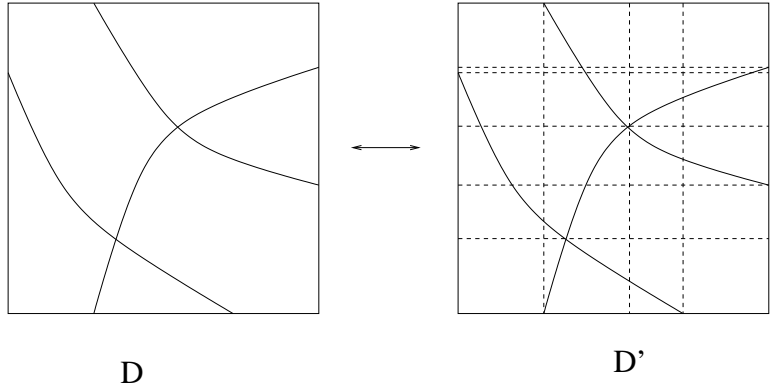}
    \caption{Decomposition of $D$.}
    \label{decompose}
    \end{center}
\end{figure}

\begin{lemma}
The discrete geodesic Fr\'echet distance between $P'$ and $Q'$ gives an $\epsilon$-approximation of the geodesic Fr\'echet distance between $P$ and $Q$, i.e. $(1-\epsilon)\delta_F(P,Q) \leq \delta_{dF}(P',Q') \leq (1+\epsilon)\delta_F(P,Q)$.
\label{disFD}
\end{lemma}

\begin{proof}
Here, $d(a,b)$ denotes the weighted length of the line segment $ab$.
Let $\alpha, \beta$ be two optimal reparametrizations between $P'$ and $Q'$. Obviously, they are also an optimal matching
between $P$ and $Q$. Following~\cite{Eiter94}, we define a coupling $L$ between $\delta(P')=\{u_1,u_2,\ldots, u_{p'}\}$ and
$\delta(Q')=\{v_1,v_2,\ldots, v_{q'} \}$.
For each point $u\in \delta(P')$, let $t(u)$ be the smallest value such that $P'(\alpha(t(u)))= u$, and for each point $v$ in
$\delta(Q')$, let $s(v)$ be the smallest value such that $Q'(\beta(s(v)))=v$. First, we add $(u_1,v_1)$ to $L$. Let $i$ and $j$
be the number of distinct points in $\delta(P')$ and $\delta(Q')$, respectively, added to $R$. If $j=q'$ or $t(u_{i+1})\le
s(v_{j+1})$, then add $(u_{i+1}, v_j)$ to $L$. We have
$$d(u_{i+1},v_j) \leq S(P'(\alpha(t(u_{i+1})))Q'(\beta(t(u_{i+1})) ) ) + d(v_j,v_{j+1})$$ and
$$S(P'(\alpha(t(u_{i+1})))Q'(\beta(t(u_{i+1})) ) ) \leq d(u_{i+1},v_j) + d(v_j,v_{j+1}).$$
If $v_jv_{j+1}$ is not contained in
any vertex vicinity, we have $$d(v_j,v_{j+1}) \leq \epsilon S(P'(\alpha(t(u_{i+1}))),Q'(\beta(t(u_{i+1})) )) \leq \epsilon
\delta_F(P,Q).$$ If $v_jv_{j+1}$ is inside a vertex vicinity, we have $$d(v_j,v_{j+1}) \leq \epsilon B \leq \epsilon\delta_F(P,Q).$$
We obtain that
$$|S(P'(\alpha(t(u_{i+1}))),Q'(\beta(t(u_{i+1})) ) ) - d(u_{i+1}, v_j) | \leq \epsilon \delta_F(P,Q).$$
If $i=p'$ or $t(u_{i+1}) > s(v_{j+1})$, then we add $(u_i,v_{j+1})$ to $L$. Similarly, we have\\
$$|S(P'(\alpha(s(v_{j+1}))),Q'(\beta(s(v_{j+1})) ) ) -  d(u_{i},v_{j+1}) |\leq \epsilon \delta_F(P,Q),$$
and thus
$$(1-\epsilon)\delta_F(P,Q) \leq \delta_{dF}(P',Q') \leq (1+\epsilon)\delta_F(P,Q)$$
\hfill $\Box$
\end{proof}

Instead of finding an exact geodesic path between two points, we can apply existing shortest path algorithms in weighted regions, which find an approximate path.
Let $\delta'_{dF}(P',Q')$ be the approximate discrete Fr\'echet distance computed by replacing the exact
shortest path algorithm by the approximation algorithm, which gives an $\epsilon$-approximation of the geodesic distance between points. We have
$$\delta'_{dF}(P',Q') \leq (1+\epsilon)\delta_dF(P,Q) \leq (1+\epsilon)^2\delta_F(P,Q) \leq (1+3\epsilon)\delta_F(P,Q).$$ 
Similarly, assuming that $\epsilon \leq 1/3$, we have  
$$\delta'_{dF}(P',Q') \ge (1-3\epsilon)\delta_F(P,Q).$$ 
Thus, we obtain
$$|\delta'_{dF}(P',Q')- \delta_F(P,Q)| \leq 3\epsilon\delta_F(P,Q).$$
If a non-query based approximation algorithm is used, $\delta'_{dF}(P',Q')$ can be computed in time
$O(p'q'T(n,\epsilon))= O(C(R)^2\frac{pq}{\epsilon^2} (\log^4 \frac{1}{\epsilon}) T(n,\epsilon))$,
where $T(n,\epsilon)$ is the time
to compute an approximate shortest path between two points.
If a query-based approximation algorithm is used, then $\delta'_{dF}(P',Q')$ can be computed in\\
$O(p'q'QUERY(n,\epsilon)+PRE(n,\epsilon))= O(C(R)^2\frac{pq}{\epsilon^2} (\log^4 \frac{1}{\epsilon}) QUERY(n,\epsilon)+PRE(n,\epsilon))$, where
$QUERY(n,\epsilon)$ is the query time and $PRE(n,\epsilon)$ is the preprocessing time of the algorithm.
For example, if we use the algorithm proposed by Aleksandrov et. al.~\cite{Ale08}, the algorithm takes $O(C(R)^2\frac{pq}{\epsilon^2} (\log^4 \frac{1}{\epsilon})\bar{q} + \frac{(g+1)n^2}{\epsilon^{2/3}\bar{q}}\log
\frac{n}{\epsilon}\log^4\frac{1}{\epsilon})$ time, where $\bar{q}$ is a query time
parameter and $g$ is the genus of the graph constructed by the discretization scheme.

\begin{theorem}
A $3\epsilon$-approximation of the Fr\'echet distance between two polygonal curves in a $2D$ weighted subdivision $R$ can be computed in $O(C(R)^2\frac{pq}{\epsilon^2} (\log^4 \frac{1}{\epsilon})\bar{q} + \frac{(g+1)n^2}{\epsilon^{2/3}\bar{q}}\log
\frac{n}{\epsilon}\log^4\frac{1}{\epsilon})$ time, where $\bar{q}$ is a query time
parameter, $g$ is the genus of the graph constructed by the discretization scheme, $C(R)$ is a constant associated with geometry of $R$, and $n$ is the number of vertices of $R$.
\end{theorem}

\subsection{Geodesic Fr\'echet Distance in $\mathcal{R}^3$ with Obstacles}

In this subsection, we briefly discuss the geodesic Fr\'echet distance problem in 1 or $\infty$ weighted regions
in $\mathcal{R}^3$ (that is, among obstacles in $\mathcal{R}^3$).
Let $R$ be a weighted subdivision in $\mathcal{R}^3$ with a total of $n$ vertices.
The weight of each region $R_i \in R$ is either 1 or $\infty$. Given two polygonal curves $P$ and $Q$ in $R$,
we want to find the Fr\'echet distance between $P$ and $Q$, where the distance between two points in $R$ is defined as the
length of the geodesic path between those points, i.e. the length of the shortest obstacle-avoiding path.

We set $r(v)=\epsilon B$ and add additional vertices on $P$ and $Q$ as described in Section 4.1. Let $P'$ and $Q'$ be the new curves.
\begin{lemma}
The discrete Fr\'echet distance between $P'$ and $Q'$ gives an $\epsilon$-approximation of the Fr\'echet distance between $P$ and $Q$, i.e. $(1-\epsilon)\delta_F(P,Q) \leq \delta_{dF}(P',Q') \leq (1+\epsilon)\delta_F(P,Q)$.
\end{lemma}

\begin{proof}
Similar to the proof of $Lemma~\ref{disFD}$.
\hfill $\Box$
\end{proof}

Let $\delta'_{dF}(P',Q')$ be the discrete Fr\'echet distance computed by a shortest path approximation algorithm,
which gives an
$\epsilon$-approximation of the shortest path. We have
$$|\delta'_{dF}(P',Q') - \delta_{F}(P,Q)| \leq 3\epsilon \delta_{F}(P,Q).$$
$\delta'_{dF}(P',Q')$ can be computed in $O(C(R)^2 pq (1/\epsilon^2) \log^4 (1/\epsilon) T(n,\epsilon)$ time,
where $C(R)$ is a constant depending on the geometry of the problem and $T(n,\epsilon)$ is the
time to approximate the shortest path between two points in $R$. For example, we can use the approximation algorithm given by Clarkson~\cite{Clar87}, which takes $O(n^2\lambda(n)\log(n/\epsilon)/{\epsilon}^4+n^2\log (n\gamma) \log(n\log\gamma))$ time, where $\gamma$ is the ratio of the length of the longest obstacle edge to the
Euclidean distance between the two points, and $\lambda(n)$ is a very slowly-growing function related to
the inverse of the Ackermann's function. Thus, our algorithm takes $O(C(R)^2 pq (1/\epsilon^2) \log^4 (1/\epsilon)(n^2\lambda(n)\log(n/\epsilon)/{\epsilon}^4+n^2\log (n\gamma) \log(n\log\gamma)))$ time.

\begin{theorem}
A $3\epsilon$-approximation of the geodesic Fr\'echet distance in $\mathcal{R}^3$ with Obstacles can be computed in $O(C(R)^2 pq (1/\epsilon^2) \log^4 (1/\epsilon)(n^2\lambda(n)$ $\log(n/\epsilon)/{\epsilon}^4+n^2\log (n\gamma) \log(n\log\gamma)))$ time, where $\gamma$ is the ratio of the length of the longest obstacle edge to the
Euclidean distance between the two points, and $\lambda(n)$ is a very slowly-growing function related to
the inverse of the Ackermann's function.
\end{theorem}

\section{Conclusion}

In this paper, we discussed three versions of the Fr\'echet distance problem in weighted regions and presented an approximation
algorithm for
each version. First, we discussed the non-geodesic Fr\'echet distance problem in planar weighted regions. We showed
that we can
approximate the Fr\'echet distance by using a parameter space, $D$, where each leash is associated with a point in $D$,
and constructing a discrete graph $G$ from $D$.
We then discussed two geodesic Fr\'echet distance problems, in planar weighted regions and in 1 or $\infty$ weighted
regions in
$\mathcal{R}^3$, and showed that in both cases, by adding additional vertices to the polygonal curves, the discrete
Fr\'echet distance can be used to approximate the continuous Fr\'echet distance.


\end{document}